\documentclass{revtex4}
\usepackage{braket}
\usepackage{amsmath}
\usepackage{amssymb}
\usepackage{amsthm}
\usepackage{graphicx}
\usepackage{calrsfs}
\DeclareMathAlphabet{\pazocal}{OMS}{zplm}{m}{n}
\theoremstyle{definition}
\newtheorem*{theorem}{Theorem}
\newtheorem{corollary}{Corollary}
\newtheorem{example}{Example}
\newtheorem{remark}{Remark}
\begin{document}
\title{Symmetric measurement-induced lower bounds of concurrence}
\author{Hao-Fan Wang}
\email{2230502117@cnu.edu.cn}
\affiliation{School of Mathematical Sciences, Capital Normal University, Beijing 100048, China}
\author{Shao-Ming Fei}
\email{feishm@cnu.edu.cn}
\affiliation{School of Mathematical Sciences, Capital Normal University, Beijing 100048, China}
\begin{abstract}
We provide a class of lower bounds for concurrence based on symmetric measurements. We show that our lower bounds estimate the quantum entanglement better than some existing lower bounds by detailed examples. Moreover, our lower bounds can be experimentally identified without state tomography.
\end{abstract}


\maketitle

\section{Introduction}
An important issue in the theory of quantum entanglement is the quantification and estimation of entanglement for composite systems. The concurrence\cite{WOS:000171609900037} is one of the most widely used entanglement measure. Nevertheless, due to the infimum involved in the calculation, analytical formulas for concurrence are obtained only for some very specific quantum states. Consequently, much attention has been paid to the estimation of the lower bounds of concurrence for general bipartite states.

Lower bounds of concurrence may be derived from separability criteria such as the positive partial transpose (PPT) criterion\cite{WOS:A1996VC33600001,WOS:A1996VV48200001} and realignment criterion\cite{WOS:000183934000001,WOS:000241549700003}. In 2005, Chen et al.\cite{WOS:000230680000010} proposed a lower bound of concurrence based on these two criteria. After that, more lower bounds have been presented from different approaches\cite{WOS:000246890400058,WOS:000288785700014,WOS:000519613200001,WOS:000576674700009,WOS:000732436000001}. In 2021, Shi et al.\cite{WOS:000940186800002} provided a lower bound of concurrence that is more general than the lower bound from realignment criterion. The result has been further developed by Lu et al.\cite{lufei24} recently.

Quantum measurements, given by positive operator-valued measures (POVMs), are used to obtain information about quantum systems. Two well known POVMs are mutually unbiased measurements (MUMs)\cite{WOS:000338940100005} and general symmetric informationally complete positive operator-valued measure (GSIC POVM)\cite{WOS:000340207200009}, which are natural extensions of measurements based on mutually unbiased bases (MUBs)\cite{WOS:A1989AG05900008} and symmetric informationally complete positive operator-valued measure (SIC POVM)\cite{WOS:000221658500005}, respectively. In fact, both GSIC POVM and a complete set of MUMs are conical 2-designs\cite{WOS:000445721600003}. Recently, the symmetric measurement or $(N,M)$-POVM has been proposed\cite{WOS:000789327100005}, which reduces MUMs and GSIC POVM for special cases.

Separability criteria based on quantum measurements are of particular significance, as they are experimentally implemental. In 2015, Chen et al.\cite{WOS:000355264700038} provided a criterion based on GSIC POVM. In 2018, based on SIC POVM Shang et al.\cite{WOS:000441330400001} presented a criterion called ESIC criterion. Lai et al.\cite{WOS:000446569700001} generalized the ESIC criterion to the one based on GSIC POVM, and shew that this criterion is better than the criterion proposed by Chen et al. In 2023, Tang et al.\cite{WOS:000996643000001} extended this result to the one based on $(N,M)$-POVM.

In this paper, we provide analytical lower bounds of concurrence with respect to the separability criteria presented in \cite{WOS:000441330400001,WOS:000446569700001,WOS:000996643000001}.
They are experimentally implemental. Detailed examples are given to show the advantages of these lower bounds.

\section{Symmetric measurement based lower bounds of concurrence}

We first recall the $(N,M)$-POVM. A set of $N$ $d$-dimensional POVMs $\{E_{\alpha,k}|k=1,2\cdots,M\}$ ($\alpha=1,2,\cdots,N$) constitute an $(N,M)$-POVM if
 \begin{flalign*}
 	{\rm tr}(E_{\alpha,k}) &= \dfrac{d}{M}, \\
 	{\rm tr}(E_{\alpha,k}^{2}) &= x,\\
 	{\rm tr}(E_{\alpha,k}E_{\alpha,l}) &= \dfrac{d-Mx}{M(M-1)},~~ l\neq k\\
 	{\rm tr}(E_{\alpha,k}E_{\beta,l}) &= \dfrac{d}{M^{2}},~~ \beta\neq\alpha
 \end{flalign*}
where the parameter $x$ satisfies $\dfrac{d}{M^{2}}<x\leq \min\left\{\dfrac{d^{2}}{M^{2}},\dfrac{d}{M}\right\}$. When $N(M-1)=d^{2}-1$, the $(N,M)$-POVM is called an informationally complete $(N,M)$-POVM.
For any finite dimension $d$ ($d>2$), there exist at least four different types of informationally complete $(N,M)$-POVM: (1) $N=1$ and $M=d^{2}$ (GSIC POVM), (2) $N=d+1$ and $M=d$ (MUMs), (3) $N=d^{2}-1$ and $M=2$, (4) $N=d-1$ and $M=d+2$.

From orthonormal Hermitian operator bases $\{G_{0}=I_{d}/\sqrt{d},\,G_{\alpha,k}|\alpha=1,\cdots,N;\,k=1,\cdots,M-1\}$ with ${\rm tr}(G_{\alpha,k})=0$, an informationally complete $(N,M)$-POVM is given by
\begin{equation*}
E_{\alpha,k}=\dfrac{1}{M}I_{d}+tH_{\alpha,k},
\end{equation*}
 where
 \begin{equation*}
 	H_{\alpha,k}=\begin{cases}
 		G_{\alpha}-\sqrt{M}(\sqrt{M}+1)G_{\alpha,k},~ k=1,\cdots,M-1\\
 		(\sqrt{M}+1)G_{\alpha}, k=M
 	\end{cases}
 \end{equation*}
with $G_{\alpha}=\sum\limits_{k=1}^{M-1}G_{\alpha,k}$. The parameter $t$ should be chosen such that $E_{\alpha,k}\geq 0$, which implies that
 \begin{equation*}
 	-\dfrac{1}{M}\dfrac{1}{\lambda_{\max}}\leq t \leq \dfrac{1}{M}\dfrac{1}{|\lambda_{\min}|},
 \end{equation*}
 where $\lambda_{\max}$ and $\lambda_{\min}$ are the minimal and maximal eigenvalues of $H_{\alpha,k}$ for all $\alpha$ and $k$, respectively. The parameters $t$ and $x$ satisfy the following relation,
 \begin{equation*}
 	x=\dfrac{d}{M^{2}}+t^{2}(M-1)(\sqrt{M}+1)^{2}.
 \end{equation*}

Let $\pazocal{H}_{1}$ and $\pazocal{H}_{2}$ be $d$-dimensional vector spaces. The concurrence of a pure state $\ket{\psi}\in \pazocal{H}_{1}\otimes \pazocal{H}_{2}$ is defined by
\begin{equation*}
	C(\ket{\psi})=\sqrt{2(1-{\rm tr}(\rho_{1}^{2}))},
\end{equation*}
where $\rho_{1}={\rm tr}_{2}(\ket{\psi}\bra{\psi})$ is the reduced state obtained by tracing over the second subsystem. The concurrence is extended to mixed states $\rho$ by the convex roof extension,
\begin{equation*}
	C(\rho)=\min\limits_{\{p_{i},\ket{\psi_{i}}\}}\sum\limits_{i}p_{i}C(\ket{\psi_{i}}),
\end{equation*}
where the minimum is taken over all possible pure state decompositions of $\rho = \sum\limits_{i}p_{i}\ket{\psi_{i}}\bra{\psi_{i}}$, with $p_{i}\geq 0$ and $\sum\limits_{i}p_{i}=1$.

\begin{theorem}
Let $\{E_{\alpha,k}|\alpha=1,\cdots,N;\,k=1,\cdots,M\}$ be a informationally complete $(N, M)$-POVM with the free parameter $x$ on the $d$ dimensional Hilbert space $\pazocal{H}$, $\rho$ be a bipartite state in $\pazocal{H}\otimes\pazocal{H}$, $\{\ket{w_{\alpha,k}}|\alpha=1,\cdots,N;\,k=1,\cdots,M\}$ be an orthonormal basis of $\mathbb{C}^{NM}$. Define $\pazocal{P}(\rho)=\sum\limits_{\alpha,\beta=1}^{N}\sum\limits_{k,l=1}^{M}{\rm tr}\left(\rho\left(E_{\alpha,k}\otimes E_{\beta,l}\right)\right)\ket{w_{\alpha,k}}\bra{w_{\beta,l}}$. The concurrence $C(\rho)$ satisfies
	\begin{equation}\label{thm}
		C(\rho)\geq \dfrac{M(M-1)}{xM^{2}-d}\sqrt{\dfrac{2}{d(d-1)}}\left(\|\pazocal{P}(\rho)\|_{\rm tr}-\dfrac{(d-1)(xM^{2}+d^{2})}{dM(M-1)}\right).
	\end{equation}
\end{theorem}
\begin{proof}
Let $\{p_i,\ket{\psi_i}\}$ be optimal pure state decomposition of $\rho$ such that $C(\rho)=\sum\limits_{i}p_{i}C(\ket{\psi_i})$. Noting that $\sum\limits_{i}p_{i}\|\pazocal{P}(\ket{\psi_{i}}\bra{\psi_{i}})\|_{\rm tr}\geq \|\pazocal{P}(\rho)\|_{\rm tr}$ by the convex property of the trace norm, we only need to prove the theorem for the case of pure states $\ket{\psi}\in \pazocal{H}\otimes \pazocal{H}$,
\begin{equation}\label{pure}
C(\ket{\psi})\geq\dfrac{M(M-1)}{xM^{2}-d}\sqrt{\dfrac{2}{d(d-1)}}
\left(\|\pazocal{P}(\ket{\psi}\bra{\psi})\|_{\rm tr}-\dfrac{(d-1)(xM^{2}+d^{2})}{dM(M-1)}\right).
\end{equation}
By Schmidt decomposition, there exist orthonormal bases $\ket{e_1},\cdots,\ket{e_{d}}$ and $\ket{f_1},\cdots,\ket{f_{d}}$ in $\pazocal{H}$ such that $\ket{\psi}=\sum\limits_{i=1}^{r}\lambda_{i}\ket{e_i}\otimes\ket{f_{i}}$, where $\lambda_{i}\geq 0$ and $\sum\limits_{i=1}^{r}\lambda_{i}^{2}=1$. Denote $\ket{u_{ij}}=\sum\limits_{\alpha=1}^{N}\sum\limits_{k=1}^{M}
\braket{e_{j}|E_{\alpha,k}|e_{i}}\ket{w_{\alpha,k}}$ and $\ket{v_{ij}}=\sum\limits_{\alpha=1}^{N}
\sum\limits_{k=1}^{M}\braket{f_{j}|E_{\alpha,k}|f_{i}}\ket{w_{\alpha,k}}$. Then
\begin{flalign*}
 & \begin{array}{ll}
 	\pazocal{P}(\ket{\psi}\bra{\psi})
 	\hspace{-2mm} & =\sum\limits_{\alpha,\beta=1}^{N}\sum\limits_{k,l=1}^{M}
 	  \sum\limits_{i=1}^{r}\sum\limits_{j=1}^{r}\lambda_{i}\lambda_{j}
 	  \braket{e_{j}|E_{\alpha,k}|e_{i}}\braket{f_{j}|E_{\beta,l}|f_{i}}\ket{w_{\alpha,k}}\bra{w_{\beta,l}}\\
 	& =\sum\limits_{i=1}^{r}\sum\limits_{j=1}^{r}\lambda_{i}\lambda_{j}
 \ket{u_{ij}}\bra{\overline{v}_{ij}},
 	\end{array}&
\end{flalign*}
where $\ket{\overline{v}_{ij}}=\sum\limits_{\beta=1}^{N}
\sum\limits_{l=1}^{M}\overline{\braket{f_{j}|E_{\beta,l}|f_{i}}}\ket{w_{\beta,l}}$ with
$\overline{\braket{f_{j}|E_{\beta,l}|f_{i}}}$ standing for complex conjugation of ${\braket{f_{j}|E_{\beta,l}|f_{i}}}$.\par

{\bf 1)} We first consider the case of $\ket{u_{ij}}=\ket{\overline{v}_{ij}}$ for $\forall i,j\in\{1,\cdots,d\}$. In this case, it is obvious that $\pazocal{P}(\ket{\psi}\bra{\psi})$ is a positive semidefinite, and then
    \begin{equation}\label{11}
    	\|\pazocal{P}(\ket{\psi}\bra{\psi})\|_{\rm tr}={\rm tr}\left(\pazocal{P}(\ket{\psi}\bra{\psi})\right)
    =\sum\limits_{i=1}^{r}\sum\limits_{j=1}^{r}\lambda_{i}\lambda_{j}\braket{u_{ij}|u_{ij}}.
    \end{equation}
We denote $F=\sum\limits_{i=1}^{d}\sum\limits_{j=1}^{d}\ket{e_{i}}\bra{e_{j}}\otimes\ket{e_{j}}\bra{e_{i}}$. Since $\sum\limits_{\alpha=1}^{N}\sum\limits_{k=1}^{M}E_{\alpha,k}\otimes E_{\alpha,k}=\dfrac{xM^{2}-d}{M(M-1)}F+\dfrac{d^{3}-xM^{2}}{dM(M-1)}I$ \cite{WOS:001270740400004,WOS:001062294000001}, we have
    \begin{flalign*}
    	& \begin{array}{ll}
    		\braket{u_{ij}|u_{i^{\prime}j^{\prime}}}
    		\hspace{-2mm} & =\sum\limits_{\alpha,\beta=1}^{N}\sum\limits_{k,l=1}^{M}
    \overline{\braket{e_{j}|E_{\alpha,k}|e_{i}}}\braket{e_{j^{\prime}}
    |E_{\beta,l}|e_{i^{\prime}}}\braket{w_{\alpha,k}|w_{\beta,l}}\\
    &
    =\sum\limits_{\alpha=1}^{N}\sum\limits_{k=1}^{M}\braket{e_{i}|E_{\alpha,k}
    |e_{j}}\braket{e_{j^{\prime}}|E_{\alpha,k}|e_{i^{\prime}}}\\
    		& =\sum\limits_{\alpha=1}^{N}\sum\limits_{k=1}^{M}\braket{e_{i}e_{j^{\prime}}
    |E_{\alpha,k}\otimes E_{\alpha,k}|e_{j}e_{i^{\prime}}}\\
    &
    =\braket{e_{i}e_{j^{\prime}}|\sum\limits_{\alpha=1}^{N}\sum\limits_{k=1}^{M}E_{\alpha,k}\otimes E_{\alpha,k}|e_{j}e_{i^{\prime}}}\\
    		& =\dfrac{xM^{2}-d}{M(M-1)}\braket{e_{i}e_{j^{\prime}}|F|e_{j}e_{i^{\prime}}}+\dfrac{d^{3}-xM^{2}}{dM(M-1)}\braket{e_{i}e_{j^{\prime}}|e_{j}e_{i^{\prime}}}\\
    		& =\dfrac{xM^{2}-d}{M(M-1)}\delta_{ii^{\prime}}\delta_{jj^{\prime}}
    +\dfrac{d^{3}-xM^{2}}{dM(M-1)}\delta_{ij}\delta_{i^{\prime}j^{\prime}}.
    	\end{array}&
    \end{flalign*}
    Therefore, $\braket{u_{ij}|u_{ij}}=\dfrac{xM^{2}-d}{M(M-1)}+\dfrac{d^{3}-xM^{2}}{dM(M-1)}\delta_{ij}$ and
    \begin{flalign}\label{12}
    	& \begin{array}{ll}
    		\|\pazocal{P}(\ket{\psi}\bra{\psi})\|_{\rm tr}
    		\hspace{-2mm} & =\dfrac{xM^{2}-d}{M(M-1)}\sum\limits_{i=1}^{r}\sum\limits_{j=1}^{r}\lambda_{i}\lambda_{j}+\dfrac{d^{3}-xM^{2}}{dM(M-1)}\sum\limits_{i=1}^{r}\sum\limits_{j=1}^{r}\delta_{ij}\lambda_{i}\lambda_{j}\\
    		& =\dfrac{xM^{2}-d}{M(M-1)}\sum\limits_{i=1}^{r}\sum\limits_{j=1}^{r}\lambda_{i}\lambda_{j}+\dfrac{d^{3}-xM^{2}}{dM(M-1)}\sum\limits_{i=1}^{r}\lambda_{i}^{2}\\
    		& =\dfrac{xM^{2}-d}{M(M-1)}\sum\limits_{i\neq j}\lambda_{i}\lambda_{j}+\dfrac{d(xM^{2}-d)+d^{3}-xM^{2}}{dM(M-1)}\sum\limits_{i=1}^{r}\lambda_{i}^{2}\\
    		& =\dfrac{(d-1)(xM^{2}+d^{2})}{dM(M-1)}+\dfrac{2(xM^{2}-d)}{M(M-1)}
    \sum\limits_{i<j}\lambda_{i}\lambda_{j}.
    	\end{array}&
    \end{flalign}

{\bf 2)} Now we prove the general case by using the above results. Note that $\braket{\overline{v}_{ij}|\overline{v}_{i^{\prime}j^{\prime}}}
=\overline{\braket{v_{ij}|v_{i^{\prime}j^{\prime}}}}$ and $\braket{v_{ij}|v_{i^{\prime}j^{\prime}}}=\dfrac{xM^{2}-d}{M(M-1)}
\delta_{ii^{\prime}}\delta_{jj^{\prime}}+\dfrac{d^{3}-xM^{2}}{dM(M-1)}
\delta_{ij}\delta_{i^{\prime}j^{\prime}}\in\mathbb{R}$. Hence, $\braket{\overline{v}_{ij}|\overline{v}_{i^{\prime}j^{\prime}}}=\braket{v_{ij}
|v_{i^{\prime}j^{\prime}}}=\braket{u_{ij}|u_{i^{\prime}j^{\prime}}}$. Thus there exists a unitary operator $U$ on $\pazocal{H}\otimes\pazocal{H}$ such that $\ket{\overline{v}_{ij}}=U\ket{u_{ij}}$ for $i,j=1,2,\cdots,d$. As a result, we have
    \begin{flalign}\label{21}
    	& \begin{array}{ll}
    		\|\pazocal{P}(\ket{\psi}\bra{\psi})\|_{\rm tr}
    		\hspace{-2mm} & =\left\|\sum\limits_{i=1}^{r}\sum\limits_{j=1}^{r}\lambda_{i}\lambda_{j}\ket{u_{ij}}\bra{\overline{v}_{ij}}\right\|_{\rm tr}=\left\|\left(\sum\limits_{i=1}^{r}\sum\limits_{j=1}^{r}\lambda_{i}\lambda_{j}\ket{u_{ij}}\bra{u_{ij}}\right)U^{\dagger}\right\|_{\rm tr}\\[14pt]
    		& =\left\|\sum\limits_{i=1}^{r}\sum\limits_{j=1}^{r}\lambda_{i}\lambda_{j}\ket{u_{ij}}\bra{u_{ij}}\right\|_{\rm tr}=\dfrac{(d-1)(xM^{2}+d^{2})}{dM(M-1)}+\dfrac{2(xM^{2}-d)}{M(M-1)}
    \sum\limits_{i<j}\lambda_{i}\lambda_{j},
    	\end{array}&
    \end{flalign}
where the last equality is from (\ref{11}) and (\ref{12}). Since $(C(\ket{\psi}))^{2}\geq \dfrac{8}{d(d-1)}\left(\sum\limits_{i<j}\lambda_{i}\lambda_{j}\right)^{2}$ \cite{WOS:000230680000010}, we have
$C(\ket{\psi})\geq 2\sqrt{\dfrac{2}{d(d-1)}}\sum\limits_{i<j}\lambda_{i}\lambda_{j}$, which gives rise to (\ref{pure}) by using (\ref{21}).
\end{proof}

\begin{remark}
The orthonormal basis $\{\ket{w_{\alpha,k}}|\alpha=1,\cdots,N;\,k=1,\cdots,M\}$ introduced in Theorem is just for statement of convenience. The conclusion of Theorem is independent of the choice of orthonormal bases. Let $\{\ket{v_{\alpha,k}}|\alpha=1,\cdots,N;\,k=1,\cdots,M\}$ be another orthonormal basis. Then there exists a unitary operator $U$ on $\pazocal{H}\otimes\pazocal{H}$ such that $\ket{w_{\alpha,k}}=U\ket{v_{\alpha,k}}$. The norm $\|\pazocal{P}(\rho)\|_{\rm tr}$ keeps invariant under basis transformations.
\end{remark}

\begin{remark}
If $\rho$ is a separable state, $C(\rho)=0$, (\ref{thm}) gives rise to that $\|\pazocal{P}(\rho)\|_{\rm tr}\leq \dfrac{(d-1)(xM^{2}+d^{2})}{dM(M-1)}$, which is consistent with the Theorem 1 in \cite{WOS:000996643000001}.
\end{remark}

We have derived a lower bound of concurrence based on the symmetric measurements. The symmetric measurements includes SIC POVM and GSIC POVM as special cases, which gave rise to separability criteria given in \cite{WOS:000441330400001} and \cite{WOS:000446569700001}, respectively. In fact, besides separability criteria, with respect to GSIC POVM and SIC POVM, we have also the following lower bounds of concurrence from our Theorem.

\begin{corollary}
  Let $\{P_1,\cdots,P_{d^2}\}$ be a general SIC POVM with the free parameter $x$ on $\pazocal{H}$, $\rho$ be a bipartite state in $\pazocal{H}\otimes\pazocal{H}$. Denote $E_k=\sqrt{\dfrac{d(d+1)}{xd^{2}+1}}P_k$($k=1,2,\cdots,d^{2}$), $p_{kl}={\rm tr}\left(\rho\left(E_k\otimes E_l\right)\right)$ and $\pazocal{P}(\rho)=\left(p_{kl}\right)_{d^{2}\times d^{2}}$. We have
  \begin{equation*}
	C(\rho)\geq \dfrac{(d-1)(xd^{2}+1)}{xd^{3}-1}\sqrt{\dfrac{2}{d(d-1)}}\left(\|\pazocal{P}(\rho)\|_{\rm tr}-\dfrac{xd^{3}-1+d(1-xd)}{(d-1)(xd^{2}+1)}\right).
  \end{equation*}
\end{corollary}

\begin{corollary}
Let $\Pi_1,\cdots,\Pi_{d^2}$ be SIC projectors. Denote $E_k=\sqrt{\dfrac{d+1}{2d}}\Pi_k$ ($k=1,2,\cdots,d^{2}$), $p_{kl}={\rm tr}\left(\rho\left(E_k\otimes E_l\right)\right)$ and $\pazocal{P}(\rho)=\left(p_{kl}\right)_{d^{2}\times d^{2}}$. Then
\begin{equation*}
	C(\rho)\geq 2\sqrt{\dfrac{2}{d(d-1)}}(\|\pazocal{P}(\rho)\|_{\rm tr}-1).
\end{equation*}
\end{corollary}

We give several examples below to illustrate the advantages of our results.

\begin{example}
	Consider the $3\otimes 3$ PPT entangled state:
	\begin{equation*}
		\rho = \dfrac{1}{4}\left(I-\sum\limits_{i=0}^{4}\ket{\psi_{i}}\bra{\psi_{i}}\right),
	\end{equation*}
    where $\ket{\psi_{0}}=\dfrac{\ket{0}(\ket{0}-\ket{1})}{\sqrt{2}}$, $\ket{\psi_{1}}=\dfrac{(\ket{0}-\ket{1})\ket{2}}{\sqrt{2}}$, $\ket{\psi_{2}}=\dfrac{\ket{2}(\ket{1}-\ket{2})}{\sqrt{2}}$, $\ket{\psi_{3}}=\dfrac{(\ket{1}-\ket{2})\ket{0}}{\sqrt{2}}$ and $\ket{\psi_{4}}=\dfrac{(\ket{0}+\ket{1}+\ket{3})(\ket{0}+\ket{1}+\ket{3})}{3}$.

Take the $(N,M)$-POVM in Theorem to be $(8,2)$-POVM with the Hermitian basis operator $G_{\alpha, k}$ given by the Gell-Mann matrices,
    $G_{11}=\dfrac{1}{\sqrt{2}}\begin{pmatrix}
    	0 & 1 & 0\\
    	1 & 0 & 0\\
    	0 & 0 & 0
    \end{pmatrix}$,
    $G_{21}=\dfrac{1}{\sqrt{2}}\begin{pmatrix}
    	0 & -\mathrm{i} & 0\\
    	\mathrm{i} & 0 & 0\\
    	0 & 0 & 0
    \end{pmatrix}$,
    $G_{31}=\dfrac{1}{\sqrt{2}}\begin{pmatrix}
    	0 & 0 & 1\\
    	0 & 0 & 0\\
    	1 & 0 & 0
    \end{pmatrix}$,
   $G_{41}=\dfrac{1}{\sqrt{2}}\begin{pmatrix}
    	0 & 0 & -\mathrm{i}\\
    	0 & 0 & 0\\
    	\mathrm{i} & 0 & 0
    \end{pmatrix}$, \\
    $G_{51}=\dfrac{1}{\sqrt{2}}\begin{pmatrix}
    	0 & 0 & 0\\
    	0 & 0 & 1\\
    	0 & 1 & 0
    \end{pmatrix}$,
    $G_{61}=\dfrac{1}{\sqrt{2}}\begin{pmatrix}
    	0 & 0 & 0\\
    	0 & 0 & -\mathrm{i}\\
    	0 & \mathrm{i} & 0
    \end{pmatrix}$,
    $G_{71}=\dfrac{1}{\sqrt{2}}\begin{pmatrix}
    	1 & 0 & 0\\
    	0 & -1 & 0\\
    	0 & 0 & 0
    \end{pmatrix}$,
    $G_{81}=\dfrac{1}{\sqrt{6}}\begin{pmatrix}
    	1 & 0 & 0\\
    	0 & 1 & 0\\
    	0 & 0 & -2
    \end{pmatrix}$.
It is verified that the parameter $x=\frac{3}{4}+t^{2}(\sqrt{2}+1)^{2}$ with $t\in[-0.2536,0.2536]$.
Setting $t=0.01$ we have $C(\rho)\geq 0.055547$ from our Theorem. By using the theorem 6 in \cite{WOS:000940186800002}, we obtain $C(\rho)\geq 0.05399$. Hence, our bound is better. Next we consider a state by mixing $\rho$ with white noise,
    \begin{equation*}
    	\rho_{p}=\dfrac{1-p}{9}I_{9}+p\rho.
    \end{equation*}
In Fig.1, the solid red line is the bound from our Theorem based on $(8,2)$-POVM with $t=0.01$, which shows that $\rho_{p}$ is entangled for $0.88218\leq p\leq 1$. The dashed bule line is the bound from Theorem 2 in \cite{lufei24} by taking $\mu=(\frac{2227}{347},\frac{4326}{571},\frac{2233}{345})^{\mathsf{T}}$ and $\nu=(\frac{6819}{1093},\frac{1580}{219},\frac{2491}{361})^{\mathsf{T}}$, which shows that $\rho_{p}$ is entangled for $0.88221\leq p\leq 1$. The dashed black line is the bound from Theorem 6 in \cite{WOS:000940186800002} with $\alpha=\beta=5$. Obviously, our lower bound is better than the above two. In addition, the solid green line is the bound from Corollary 1 based on GSIC POVM with $x=0.04984$, the solid purple line is the bound from Corollary 2, the solid orange line is the bound from realignment\cite{WOS:000230680000010}, $C(\rho)\geq \sqrt{\dfrac{2}{d(d-1)}}\left(\|\pazocal{R}(\rho)\|_{\rm tr}-1\right)$, where $\pazocal{R}(\rho)$ is realigned matrix of $\rho$.
\begin{figure}[t]
    	\includegraphics[width=10cm]{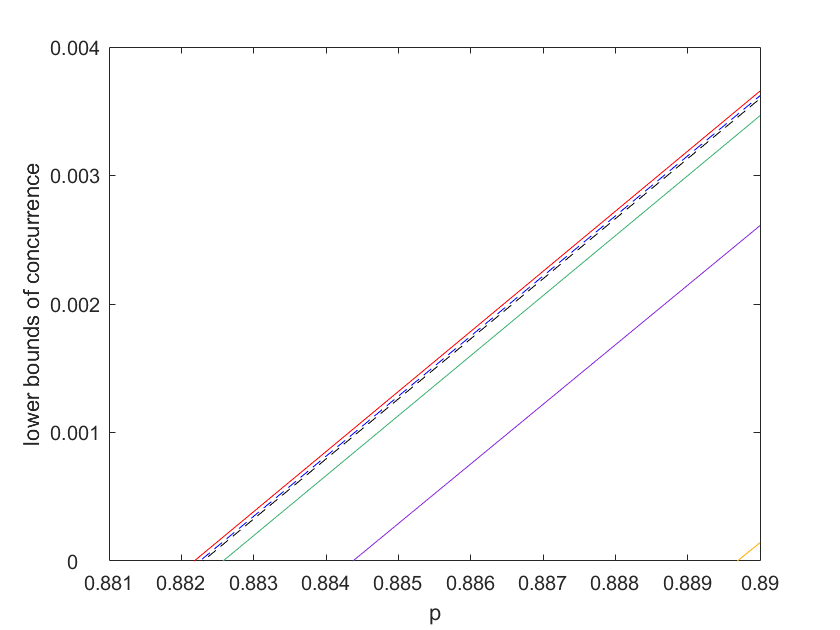}
    	\caption{Lower bounds of $C(\rho_{p})$. Solid red line for the bound from our Theorem, dashed bule line for the bound from Theorem 2 in \cite{lufei24}, dashed black line is the bound from Theorem 6 in \cite{WOS:000940186800002}, the solid green line is the bound from Corollary 1 based on GSIC POVM with $x=0.04984$, the solid purple line is the bound from Corollary 2, the solid orange line is the bound from realignment.}
    \end{figure}
\end{example}

\begin{example}
    Consider the following isotropic state,
    \begin{equation*}    	\rho_{f}=\dfrac{d^{2}}{d^{2}-1}\left(\dfrac{1-f}{d^{4}}I
    +(f-\dfrac{1}{d^{2}})\ket{\Psi^{+}}\bra{\Psi^{+}}\right)
    =v\ket{\Psi^{+}}\bra{\Psi^{+}}+(1-v)\dfrac{I}{d^{2}},
    \end{equation*}
    where $I$ is the identity operator on ${H\otimes H}$, $\ket{\Psi^{+}}=\dfrac{1}{\sqrt{d}}\sum\limits_{i=1}^{d}\ket{ii}$ with $0\leq f\leq 1$ and $v=\dfrac{d^{2}f-1}{d^{2}-1}$.
    Recall that $(\mathbb{I}_{d}\otimes Q)\ket{\Psi^{+}}=(Q^{\mathsf{T}}\otimes \mathbb{I}_{d})\ket{\Psi^{+}}$ for any $Q\in L(H)$. Then
    \begin{flalign*}
    	& \begin{array}{ll}
    		{\rm tr}\left(\rho_{f}\left(E_{\alpha,k}\otimes E_{\beta,l}^{\ast}\right)\right)
    		\hspace{-2mm} & ={\rm tr}\left(\left(v\ket{\Psi^{+}}\bra{\Psi^{+}}+(1-v)\dfrac{I}{d^{2}}\right)\left(E_{\alpha,k}\otimes E_{\beta,l}^{\ast}\right)\right)\\
    		& =v{\rm tr}((E_{\alpha,k}\otimes E_{\beta,l}^{\ast})\ket{\Psi^{+}}\bra{\Psi^{+}})+\dfrac{1-v}{d^{2}}{\rm tr}(E_{\alpha,k}\otimes E_{\beta,l}^{\ast})\\[10pt]
    		& =v{\rm tr}((E_{\alpha,k}E_{\beta,l}^{\dagger}\otimes \mathbb{I}_{d})\ket{\Psi^{+}}\bra{\Psi^{+}})+\dfrac{1-v}{d^{2}}{\rm tr}(E_{\alpha,k}){\rm tr}(E_{\beta,l}^{\ast})\\[5pt]
    		& =\dfrac{v}{d}{\rm tr}(E_{\alpha,k}E_{\beta,l})+(1-v)\dfrac{1}{M^{2}}.
    	\end{array}&
    \end{flalign*}
Denote $A=\dfrac{v(M^{2}x-d)}{dM(M-1)}I_{M}+\left(\dfrac{1}{M^{2}}
-\dfrac{v(M^{2}x-d)}{dM^{2}(M-1)}\right)J_{M}$, where $J_{M}$ is the matrix whose entries are all one. As a result, we obtain
    \begin{equation*}
    	\|\pazocal{P}(\rho)\|_{\rm tr}=N{\rm tr}(A)=\dfrac{N}{M}+\dfrac{vN(M^{2}x-d)}{dM}.
    \end{equation*}
    By using our Theorem, we have
    \begin{flalign*}
    	& \begin{array}{ll}
    		C(\rho)
    		\hspace{-2mm} & \geq \sqrt{\dfrac{2}{d(d-1)}}\dfrac{v(d^{2}-1)(M^{2}x-d)+d(d^{2}-1)-d(xM^{2}-d)-d^{3}+xM^{2}}{d(xM^{2}-d)}\\[10pt]
    		& \quad =\sqrt{\dfrac{2}{d(d-1)}}\left(\dfrac{d^{2}f-1-d}{d}+\dfrac{xM^{2}-d}{d(xM^{2}-d)}\right)\\[10pt]
    		& \quad =\sqrt{\dfrac{2}{d(d-1)}}\left(df-1\right).
    	\end{array}&
    \end{flalign*}
Interestingly, for the isotropic state our bound is the exact value of the concurrence\cite{WOS:000180804600033}.
\end{example}

\begin{example}
	Consider the mixture of the bound entangled state proposed by Horodecki \cite{WOS:A1997XN80700002},
	\begin{equation*}
		\rho_{\tau}=\dfrac{1}{1+8\tau}
		\begin{pmatrix}
			\tau & 0 & 0 & 0 & \tau & 0 & 0 & 0 & \tau\\
			0 & \tau & 0 & 0 & 0 & 0 & 0 & 0 & 0\\
			0 & 0 & \tau & 0 & 0 & 0 & 0 & 0 & 0\\
			0 & 0 & 0 & \tau & 0 & 0 & 0 & 0 & 0\\
			\tau & 0 & 0 & 0 & \tau & 0 & 0 & 0 & \tau\\
			0 & 0 & 0 & 0 & 0 & \tau & 0 & 0 & 0\\
			0 & 0 & 0 & 0 & 0 & 0 & \frac{1+\tau}{2} & 0 & \frac{\sqrt{1-\tau^{2}}}{2}\\
			0 & 0 & 0 & 0 & 0 & 0 & 0 & 0 & 0\\
			\tau & 0 & 0 & 0 & \tau & 0 & \frac{\sqrt{1-\tau^{2}}}{2} & 0 & \frac{1+\tau}{2}
		\end{pmatrix}
	\end{equation*}
    and the $9\times 9$ identity matrix $I_{9}$,
    \begin{equation*}
    	\rho(\tau,q)=q\rho_{\tau}+\frac{1-q}{9}I_{9}.
    \end{equation*}
    Fig.2 illustrate the lower bounds of $C(\rho(\tau,q))$ for $q=0.995$. The red line is the bound from our Theorem based on the $(8,2)$-POVM with $t=0.01$. The green line is the bound from Corollary 1 based on GSIC POVM with $x=0.04984$. The purple line is the bound from Corollary 2. The orange line is the bound from realignment. It is seen that lower bound from realignment fails to detect the entanglement of $\rho(\tau,0.995)$, while our lower bounds still works in identifying the entanglement.
    \begin{figure}[t]
    	\includegraphics[width=10cm]{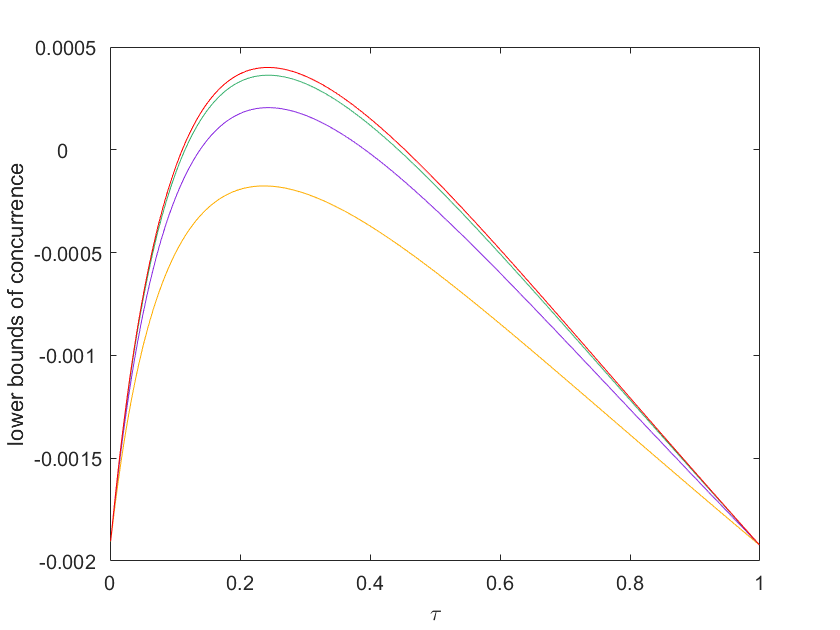}
    	\caption{Lower bounds of $C(\rho(\tau,0.995))$. Red line for the bound from our Theorem, green line for the bound from Corollary 1, purple line for the bound from Corollary 2, orange line for the bound from realignment.}
    \end{figure}
\end{example}

\begin{example}
	Let us consider the following state given in \cite{WOS:000290231700012},
	\begin{equation*}
		\rho = \dfrac{1}{4}{\rm diag}\{q_{1},q_{4},q_{3},q_{2},q_{2},q_{1},q_{4},q_{3},q_{3},
q_{2},q_{1},q_{4},q_{4},q_{3},q_{2},q_{1}\}+\dfrac{q_{1}}{4}\sum\limits_{i,j=1,6,11,16}^{i\neq j}F_{ij},
	\end{equation*}
    where $F_{ij}$ is a matrix with the $(i,j)$ entry $1$ and the rest of the entries $0$, $q_{m}\geq 0$ and $\sum\limits_{m=1}^{4}q_{m}=1$.

We construct $(5,4)$-POVM with the Hermitian basis operators given by the following general Gell-Mann matrices,\\
    $G_{11}=\dfrac{1}{\sqrt{2}}\begin{pmatrix}
    	0 & -\mathrm{i} & 0 & 0 \\
    	\mathrm{i} & 0 & 0 & 0 \\
    	0 & 0 & 0 & 0 \\
    	0 & 0 & 0 & 0
    \end{pmatrix}$,~
    $G_{12}=\dfrac{1}{\sqrt{2}}\begin{pmatrix}
    	0 & 0 & -\mathrm{i} & 0 \\
    	0 & 0 & 0 & 0 \\
    	\mathrm{i} & 0 & 0 & 0 \\
    	0 & 0 & 0 & 0
    \end{pmatrix}$,~
    $G_{13}=\dfrac{1}{\sqrt{2}}\begin{pmatrix}
    	0 & 0 & 0 & -\mathrm{i} \\
    	0 & 0 & 0 & 0 \\
    	0 & 0 & 0 & 0 \\
    	\mathrm{i} & 0 & 0 & 0
    \end{pmatrix}$, \\
    $G_{21}=\dfrac{1}{\sqrt{2}}\begin{pmatrix}
    	0 & 1 & 0 & 0 \\
    	1 & 0 & 0 & 0 \\
    	0 & 0 & 0 & 0 \\
    	0 & 0 & 0 & 0
    \end{pmatrix}$,~
    $G_{22}=\dfrac{1}{\sqrt{2}}\begin{pmatrix}
    	0 & 0 & 0 & 0 \\
    	0 & 0 & -\mathrm{i} & 0 \\
    	0 & \mathrm{i} & 0 & 0 \\
    	0 & 0 & 0 & 0
    \end{pmatrix}$,~
    $G_{23}=\dfrac{1}{\sqrt{2}}\begin{pmatrix}
    	0 & 0 & 0 & 0 \\
    	0 & 0 & 0 & -\mathrm{i} \\
    	0 & 0 & 0 & 0 \\
    	0 & \mathrm{i} & 0 & 0
    \end{pmatrix}$, \\
    $G_{31}=\dfrac{1}{\sqrt{2}}\begin{pmatrix}
    	0 & 0 & 1 & 0 \\
    	0 & 0 & 0 & 0 \\
    	1 & 0 & 0 & 0 \\
    	0 & 0 & 0 & 0
    \end{pmatrix}$,~
    $G_{32}=\dfrac{1}{\sqrt{2}}\begin{pmatrix}
    	0 & 0 & 0 & 0 \\
    	0 & 0 & 1 & 0 \\
    	0 & 1 & 0 & 0 \\
    	0 & 0 & 0 & 0
    \end{pmatrix}$,~
    $G_{33}=\dfrac{1}{\sqrt{2}}\begin{pmatrix}
    	0 & 0 & 0 & 0 \\
    	0 & 0 & 0 & 0 \\
    	0 & 0 & 0 & -\mathrm{i} \\
    	0 & 0 & \mathrm{i} & 0
    \end{pmatrix}$, \\
    $G_{41}=\dfrac{1}{\sqrt{2}}\begin{pmatrix}
    	0 & 0 & 0 & 1 \\
    	0 & 0 & 0 & 0 \\
    	0 & 0 & 0 & 0 \\
    	1 & 0 & 0 & 0
    \end{pmatrix}$,~
    $G_{42}=\dfrac{1}{\sqrt{2}}\begin{pmatrix}
    	0 & 0 & 0 & 0 \\
    	0 & 0 & 0 & 1 \\
    	0 & 0 & 0 & 0 \\
    	0 & 1 & 0 & 0
    \end{pmatrix}$,~
    $G_{43}=\dfrac{1}{\sqrt{2}}\begin{pmatrix}
    	0 & 0 & 0 & 0 \\
    	0 & 0 & 0 & 0 \\
    	0 & 0 & 0 & 1 \\
    	0 & 0 & 1 & 0
    \end{pmatrix}$, \\
    $G_{51}=\dfrac{1}{\sqrt{2}}\begin{pmatrix}
    	1 & 0 & 0 & 0 \\
    	0 & -1 & 0 & 0 \\
    	0 & 0 & 0 & 0 \\
    	0 & 0 & 0 & 0
    \end{pmatrix}$,~
    $G_{52}=\dfrac{1}{\sqrt{6}}\begin{pmatrix}
    	1 & 0 & 0 & 0 \\
    	0 & 1 & 0 & 0 \\
    	0 & 0 & -2 & 0 \\
    	0 & 0 & 0 & 0
    \end{pmatrix}$,~
    $G_{53}=\dfrac{1}{2\sqrt{3}}\begin{pmatrix}
    	1 & 0 & 0 & 0 \\
    	0 & 1 & 0 & 0 \\
    	0 & 0 & 1 & 0 \\
    	0 & 0 & 0 & -3
    \end{pmatrix}$.\\
It is verified that the corresponding parameter $x=\frac{1}{4}+27t^{2}$, $t\in [-0.0572,0.0680]$.

Set $q_{4}=-\frac{1}{3}q_{1}+\frac{1}{2}$, $q_{2}=q_{3}=\frac{1-q_{1}-q_{4}}{2}$ and $t=0.01$. Fig.3 shows the bound from our Theorem (solid red line), the bound from Theorem 2 in \cite{WOS:001237202800008} (dashed bule line) with $z=1$.
    \begin{figure}[t]
    	\includegraphics[width=10cm]{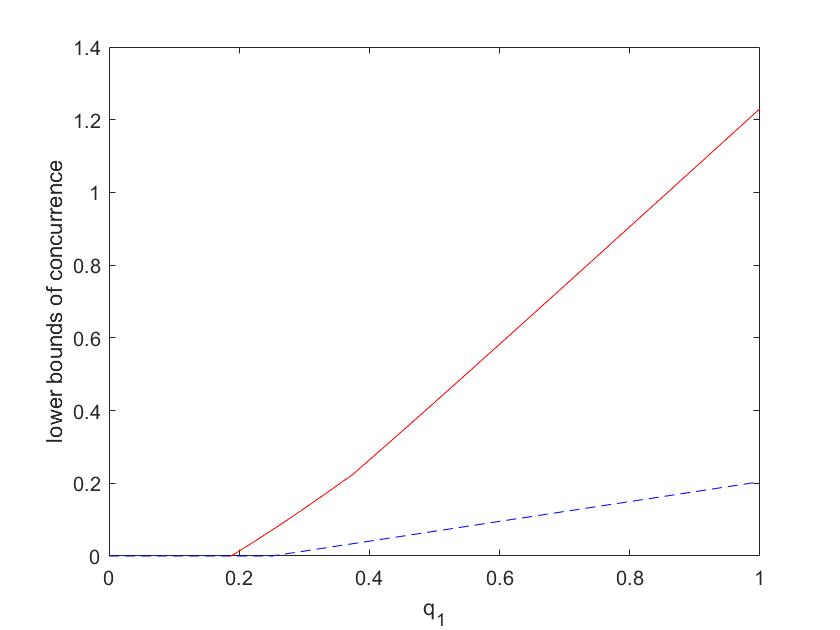}
    	\caption{Lower bounds of $C(\rho)$. Solid red line for the bound from our Theorem, dashed bule line for the bound from Theorem 2 in \cite{WOS:001237202800008} with $z=1$.}
    \end{figure}
\end{example}

\section{Conclusions and discussions}
We have derived a class of lower bounds for concurrence by using symmetric measurements. The
lower bounds are shown to be better than some related ones by detailed examples. More importantly, our bounds can be experimentally implemented. It not necessary to carry out tomography in order to estimate the bounds.

In \cite{WOS:000441330400001} the authors presented the following conjecture: Given a bipartite state $\rho$ in $\pazocal{H}\otimes\pazocal{H}$, if $\|\pazocal{R}(\rho)\|_{\rm tr}>1$, then $\|\pazocal{P}(\rho)\|_{\rm tr}>1$, where $\pazocal{P}(\rho)$ is the particular one defined in our Corollary 2. This conjecture was proved to be true in \cite{WOS:000791874300001}. Based on this conjecture and our theorem, one may conjecture that for any bipartite state $\rho$ in $\pazocal{H}\otimes\pazocal{H}$,
	\begin{equation*}
		\dfrac{M(M-1)}{xM^{2}-d}\left(\|\pazocal{P}(\rho)\|_{\rm tr}-\dfrac{(d-1)(xM^{2}+d^{2})}{dM(M-1)}\right)\geq \|\pazocal{R}(\rho)\|_{\rm tr}-1.
	\end{equation*}
In addition, since the proof of our Theorem mainly relies on the fact that any informationally complete $(N,M)$-POVM is a conical 2-design, it would be possible to obtain new better lower bounds of concurrence based on certain other kinds of conical 2-designs by using the method similar to the proof of our Theorem. We believe that investigating lower bounds of concurrence based on conical 2-designs may be an intriguing endeavor.

\bigskip
\noindent{\bf Acknowlegements} This work is supported by the National Natural Science Foundation of China (NSFC) under Grants 12075159 and 12171044, and the specific research fund of the Innovation Platform for Academicians of Hainan Province.

\end{document}